\documentclass[journal,final,onecolumn]{IEEEtran}
\usepackage{fixltx2e}
\usepackage{cite}
\usepackage{url}
\usepackage{mathrsfs}
\usepackage{color}
\usepackage{float}
\usepackage[caption=false]{subfig}

\ifCLASSINFOpdf
   \usepackage[pdftex]{graphicx}
   \graphicspath{{Figs/}}
   \DeclareGraphicsExtensions{.pdf,.jpeg,.png}
\else
\fi
\usepackage{balance}
\usepackage[cmex10]{amsmath}
\usepackage{amsmath}
\usepackage{amssymb}
\usepackage{amsthm}
\usepackage{amsfonts}
\usepackage{bm}
\usepackage{xfrac}
\usepackage{empheq}
\usepackage[normalem]{ulem} 
\usepackage{soul} 
\usepackage{mathtools}

\newtheorem{theorem}{Theorem}

\newtheorem{proposition}{Proposition}
\newtheorem{lemma}{Lemma}

\newtheorem{corollary}{Corollary}
\newtheorem{remark}{Remark}
\theoremstyle{definition}

\usepackage[a4paper,bindingoffset=0.2in,%
left=1in,right=1in,top=1in,bottom=1in,%
footskip=.25in]{geometry}
\graphicspath{{figs/}}

\begin{document}
	\newgeometry{left=14mm,right=14mm,top=30mm,bottom=30mm}
	\title{On Information Theoretic Fairness With A Bounded Point-Wise Statistical Parity Constraint: An Information Geometric Approach}
\author{
	\IEEEauthorblockN{Amirreza Zamani$^\dagger$, Ayfer \"{O}zg\"{u}r$^\ddagger$, Mikael Skoglund$^\dagger$ 
		\IEEEauthorblockA{\\
			$^\dagger$Division of Information Science and Engineering, KTH Royal Institute of Technology \\
			$^\ddagger$Department of Electrical Engineering, Stanford\\
			Email: \protect amizam@kth.se,  aozgur@stanford.edu, skoglund@kth.se }}
}
	\maketitle

\begin{abstract}
In this paper, we study an information-theoretic problem of designing a fair representation under a bounded point-wise statistical (demographic) parity constraint.
	More specifically, an agent uses some useful data (database) $X$ to solve a task $T$. Since both $X$ and $T$ are correlated with some latent sensitive attribute or secret $S$, the agent designs a representation $Y$ that satisfies a bounded point-wise statistical parity, that is, such that for all realizations of the representation $y\in\cal Y$, we have $\chi^2(P_{S|y};P_S)\leq \epsilon$. In contrast to our previous work, here we use the point-wise measure instead of a bounded mutual information, and we assume that the agent has no direct access to $S$ and $T$; hence, the Markov chains $S - X - Y$ and $T - X - Y$ hold. In this work, we design $Y$ that maximizes the mutual information $I(Y;T)$ about the task while satisfying a bounded compression rate constraint, that is, ensuring that $I(Y;X) \leq r$. Finally, $Y$ satisfies the point-wise bounded statistical parity constraint $\chi^2(P_{S|y};P_S)\leq \epsilon$. 
    
	When $\epsilon$ is small, concepts from information geometry allow us to locally approximate the KL-divergence and mutual information. To design the representation $Y$, we utilize this approximation and show that the main complex fairness design problem can be rewritten as a quadratic optimization problem that has simple closed-form solution under certain constraints. For the cases where the closed-form solution is not obtained we obtain lower bounds with low computational complexity. Here, we provide simple fairness designs with low complexity which are based on finding the maximum singular value and singular vector of a matrix.
	Finally, in a numerical example we compare our obtained results with the optimal solution. %
\end{abstract}
\section{Introduction}
\begin{figure}[t]
	\centering
	\includegraphics[width = 0.45\textwidth]{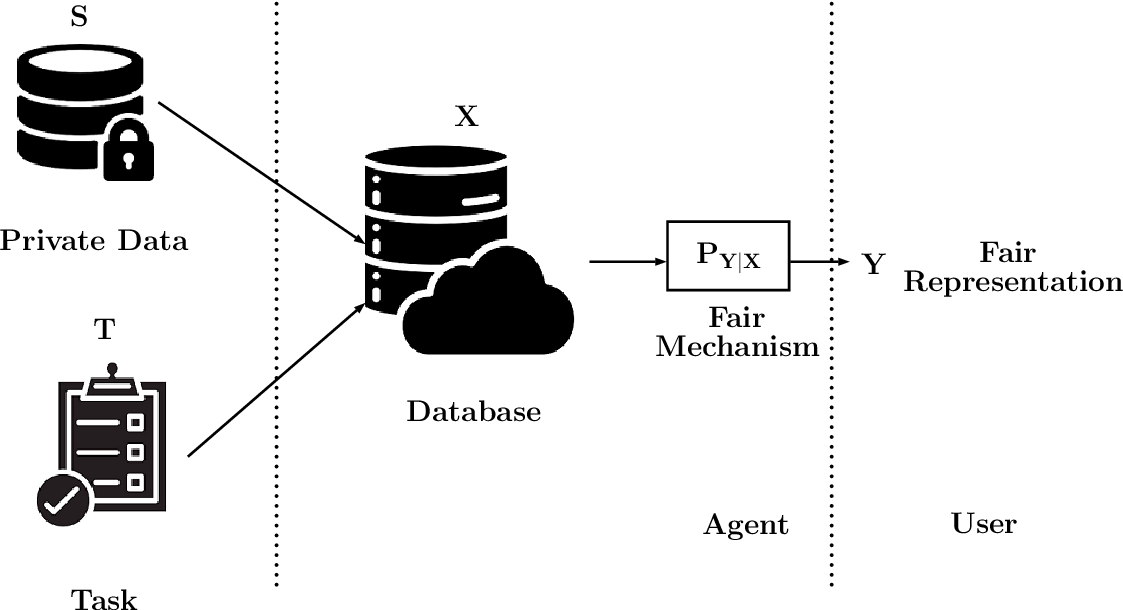}
	\caption{Fair mechanism design under bounded point-wise statistical parity. The goal is to design $Y$ using database $X$ that is useful for the task $T$, is compressed, and leaks within a controlled threshold of the latent sensitive attribute $S$. Here, the agent has no direct access to $S$ and $T$.}
	\label{IZS}
\end{figure}
As shown in Fig.~\ref{IZS}, in this paper, an agent wants to
use some observable useful data (database) $X$ to draw
inferences or make some decision about a task $T$. Here, we
assume that both the useful data and the task are correlated
with some latent sensitive attribute or secret $S$. For example,
the data $X$ could represent a person’s resume, the task $T$
could determine whether this person should be employed
in a position or not, and the sensitive attribute $S$ could
correspond to the person’s gender or ethnicity.

It is important to ensure that decisions are not unfair and that inferences do not lead to privacy breaches \cite{AmirITW2024,vari}. To address this, one possible solution is to design a \emph{private} or \emph{fair} representation $Y$ of the database $X$. This representation should preserve as much information as possible about the task $T$ while satisfying fairness or privacy criteria with respect to the latent sensitive attribute or secret
\cite{AmirITW2024,gun, vari,zhao2022, zhao2019,zemel, gronowski2023classification,hardt, king3, borz, khodam, Khodam22,Yanina1,makhdoumi, yamamoto1988rate, sankar, Calmon2,fairAsoodeh,9518124,isik2023exact,liu2024universal}. As outlined in~\cite{vari,AmirITW2024,fairAsoodeh}, privacy and fairness problems are very similar, and their intersections can be utilized in broad areas such as classification algorithms. More specifically, in \cite{AmirITW2024}, the design of representations $Y$ that contain no information about the sensitive attribute $S$ yielding $I(Y;S) = 0$, has been considered. 
In the \emph{information-theoretic privacy} literature, the independence of disclosed data $Y$ and private attribute or secret $S$ is referred to as \emph{perfect privacy} or \emph{perfect secrecy}~\cite{borz, Yanina1}. In the \emph{algorithmic fairness} literature, the independence of an algorithm’s output (representation of the database) $Y$ from the sensitive attribute $S$ is known as \emph{perfect demographic (or statistical) parity}~\cite{AmirITW2024,vari,zhao2022,zhao2019,zemel}.
As discussed in \cite{11123370,shannon,khodam,Yanina1,fairAsoodeh,king3}, there are many cases where perfect privacy (secrecy) or perfect demographic (statistical) parity is not attainable (for example, \cite[Example 1]{11123370}), hence, in \cite{11123370}, a bounded statistical parity criteria is used, that is $I(S;Y)\leq \epsilon$. In this work, instead of bounded mutual information, we use bounded point-wise statistical parity, as in \cite{khodam} for a privacy design problem. More specifically, we use $\chi^2(P_{S|y};P_S)\leq \epsilon$ as the fairness criterion which must hold for each $y\in \cal Y$. 
Using a point-wise measure can be motivated by the definition of bounded statistical (demographic) parity, which is based on the difference between the conditional distributions of the outcome given the sensitive information \cite{feldman2015certifying,fairAsoodeh,5360534}.
In contrast to \cite{AmirITW2024} and \cite{11123370}, here we assume that the latent sensitive attribute or secret $S$ and the task $T$ are not directly accessible to the agent, as shown in Fig.~\ref{IZS}. In other words, here, we assume that the Markov chains $S-X-Y$ and $T-X-Y$ hold.

Many information theory problems face challenges due to the lack of a geometric structure in the space of probability distributions. Assuming the distributions of interest are close, we can approximate the KL-divergence as well as the mutual information by a weighted squared Euclidean distance, leading to a method where the main complex problems can be approximated. This approach has been used in \cite{Shashi, huang} for point-to-point channels and certain broadcast channels. Furthermore, similar techniques have been used in the privacy setting in \cite{khodam, Khodam22, king3, 9965856, emma,8125176,zamani2025information}. 
Under the outlined fairness requirements of the representations, we consider a problem that solves a trade-off between utility, compression rate, and bounded point-wise statistical parity:
    To ensure that the representation $Y$ has as much impact as possible, we maximize the information it contains about the task $T$. Moreover, we impose a minimum level of compression $r$ to the output representation, that is, $I(Y;X) \leq r$. 
\textbf{Prior works on fairness mechanism design:}
The concept of \emph{fair representations} was introduced in \cite{zemel}, marking a significant advancement in the field of algorithmic fairness through the capabilities of deep learning. Subsequent research has been predominantly characterized by adversarial~\cite{zemel, edwards2016censoring, zhao2019} and variational~\cite{vari, creager2019flexibly, louizos2015variational, gupta2021controllable} methodologies. The theoretical exploration of the trade-offs between utility and fairness is further studied in~\cite{zhao2022}.
In \cite{vari}, the authors study an information-theoretic fairness problem. They introduce the CFB (Conditional Fairness Bottleneck), which quantifies the trades-off between the utility, fairness, and \emph{compression} of the representations in terms of mutual information. This additional criterion has since been used in subsequent studies~\cite{gupta2021controllable, de2022funck, gronowski2023classification}.~\cite{gupta2021controllable, de2022funck, gronowski2023classification}. More specifically, in \cite{vari}, the representation design is based on a variational approach. In contrast with \cite{vari}, in \cite{AmirITW2024}, the authors propose a simple and constructive theoretical approach to design fair representations with perfect demographic (statistical) parity. They obtain upper and lower bounds on a trade-off between the utility, fairness, and compression rate of the representation. To achieve the lower bounds, they use randomization techniques and extended versions of \emph{Functional Representation Lemma} and \emph{Strong Functional Representation Lemma} derived in \cite{king3}. 
In \cite{11123370}, the problem in \cite{AmirITW2024} has been extended by replacing zero demographic parity by using a bounded mutual information as the fairness constraint. 
In \cite{fairAsoodeh}, a binary classification problem subject to both differential privacy and fairness constraints is studied. In \cite{9517766}, the role of data processing in the fairness–accuracy trade-off, using equalized odds as the fairness criterion, has been studied.\\
\textbf{Prior works on privacy mechanism design:}
In \cite{Yanina1}, the authors introduce the notion of \emph{secrecy by design} and apply it to privacy mechanism design and lossless compression problems. They find bounds on the trade-off between perfect privacy and utility using the \emph{Functional Representation Lemma} (FRL). In \cite{king3}, these results are generalized to an arbitrary privacy leakage. To attain privacy mechanism, extensions of the FRL and SFRL have been used. In \cite{zamani2025variable}, a compression design problem under the secrecy constraints is studied and bounds on the average length of the encoded message are derived. 
A similar approach to that of the present paper was used in \cite{zamani2025information}, where an information-geometric perspective on Local Information Privacy (LIP), Max-lift, and Local Differential Privacy (LDP) was considered.
In \cite{borz}, the authors study the privacy-utility trade-off under perfect privacy condition. In \cite{khodam, Khodam22}, this is generalized by relaxing the perfect privacy constraint and allowing some small bounded leakage. In both \cite{khodam} and \cite{Khodam22}, point-wise (per-letter) measures have been used in the leakage constraint. Furthermore, the privacy design is based on the information geometry concept which allows them to study the problem geometrically.
The notion of the \emph{Privacy Funnel} is introduced in~\cite{makhdoumi}, where the privacy-utility trade-off considers the log-loss as the privacy and distortion metrics. 
In \cite{Calmon2}, the authors studied the fundamental limits of the privacy and utility trade-off under an estimation-theoretic approach. In \cite{yamamoto, sankar}, the trade-off between privacy and utility is studied considering the equivocation as the privacy measure and the expected distortion as a utility.\\
\textbf{Contributions:}
In this paper, we propose a simple and constructive theoretical approach to design fair representations under bounded point-wise statistical parity.
We study the problem of finding a trade-off between utility, fairness, and compression rate, which is described in~\eqref{problem}. As outlined earlier, we use a point-wise measure for the fairness constraint. We show that when fairness threshold is small, the main optimization problem can be approximated by a quadratic one. This leads to a simple design based on finding the maximum singular value and the corresponding singular vector of a matrix. We compare the obtained results with the optimal solution in a numerical example. \vspace{-1mm}
\section{System model and Problem Formulation} \label{sec:system}
In this section, we introduce the problem of designing a compressed representation of data under bounded point-wise fairness constraint. Let the latent sensitive attribute $S$, database (useful data) $X$, and task $T$ be discrete random variables (RVs) defined on alphabets $\mathcal{S}$, $\mathcal{X}$, and $\mathcal{T}$, of finite cardinalities $\left| \mathcal{S} \right|
$, $\left| \mathcal{X} \right|$, and $\left| \mathcal{T} \right|
$, respectively. The marginal probability distributions of $S$, $X$, and $T$ are denoted by $P_S \in \mathbb{R}^{|\mathcal{S}|}$, $P_X \in \mathbb{R}^{|\mathcal{X}|}$, and $P_T \in \mathbb{R}^{|\mathcal{T}|}$, respectively. Furthermore, $P_{S,X,T} \in \mathbb{R}^{|\mathcal{S}|\times |\mathcal{X}| \times |\mathcal{T}|}$ denotes the joint distribution of $S$, $X$ and $T$. 
Let the matrices $P_{S|X}\in \mathbb{R}^{|\mathcal{S}\times|\mathcal{X}|}$ and $P_{T|X}\in \mathbb{R}^{|\mathcal{T}\times|\mathcal{X}|}$ with elements $P_{S=s|X=x}$ and $P_{T=t|X=x}$ denote the conditional distributions of $S$ given $X$ and $T$ given $X$.
Here, for simplicity of design, we assume that $P_{S|X}$ is invertible, implying $|\mathcal{X}| = |\mathcal{S}|$. However, we will later discuss how this assumption can be generalized.
Let $\sqrt{P_S} \in \mathbb{R}^{|\mathcal{S}|}$ denote a vector with entries $\sqrt{P_{S}(s)}$, and let $[\sqrt{P_S}] \in \mathbb{R}^{|\mathcal{S}|\times|\mathcal{S}|}$ denote a diagonal matrix with diagonal elements $\sqrt{P_S(s)}$. We use the same notation for $[\sqrt{P_{T}}]$ and $[\sqrt{P_X}]$.
Let the discrete RV $Y \in \mathcal{Y}$ denote the fair representation. We emphasize that here we assume that the Markov chains $S-X-Y$ and $T-X-Y$ hold, as the agent has no direct access to $S$ and $T$.  
The goal is to design a mapping $P_{Y|X} \in \mathbb{R}^{|\mathcal{Y}| \times |\mathcal{X}|}$ that maximizes the information it keeps about the task $T$, while maintaining a minimum level of compression $r$ and satisfying the bounded point-wise statistical parity, that is $I(Y;X)\leq r$ and for all $y\in\cal Y,$ we have $\chi^2(P_{S|y};P_S)\leq \epsilon$. More specifically, the representation $Y$ should satisfy the fairness constraint defined as
\begin{align}
    \chi^2(P_{S|y};P_S)=\|\frac{P_{S|y}(s)-P_{S}(s)}{P_S(s)}\|^2=\sum_s\left(\frac{P_{S|y}(s)-P_{S}(s)}{P_S(s)}\right)^2\leq \epsilon^2,\ \forall y,\label{local1}
\end{align}
where $\|\cdot\|$ corresponds to the Euclidean norm. 
Intuitively, for small $\epsilon$, \eqref{local1} means that the two distributions (vectors) $P_{S|Y=y}$ and $P_S$ are close to each other.
The closeness of $P_{S|Y=y}$ and $P_S$ allows us to locally approximate $I(T;Y)$ and $I(X;Y)$ which leads to an approximation of \eqref{problem}. For more details on the point-wise measure, see \cite{khodam,Khodam22}.
\begin{remark}
In \cite{9965856}, we used bounded mutual information as the statistical parity constraint. Furthermore,
    bounded statistical parity has been used in \cite{feldman2015certifying,fairAsoodeh,5360534}, where we have 
    \begin{align}\label{r}
        |P(Y=y|S=\!s_1)-P(Y=y|S=s_2)|\leq\epsilon, \forall s_1,s_2 \in \mathcal{S}.
    \end{align}
    Here, we use a similar but slightly different measure as defined in \eqref{local1}. Using a point-wise measure can be motivated by \eqref{r}. The relation between \eqref{local1} and bounded mutual information has been studied in \cite[Remark 3]{khodam}. 
\end{remark}
Hence, the fair representation design problem can be stated as follows:
\begin{align}
    g^{r}_{\epsilon,\chi^2}(P_{S,X,T})
    &\triangleq\sup_{\begin{array}{c} 
	\substack{P_{Y|X}:S-X-Y,\ T-X-Y,\\ \chi^2(P_{S|y};P_S)\leq \epsilon^2,\ \forall y,\\ I(X;Y)\leq r }
	\end{array}}I(Y;T).
    \label{problem}
\end{align} 
The constraint $\chi^2(P_{S|y};P_S)\leq \epsilon^2,\ \forall y$, corresponds to the bounded point-wise statistical parity, and $I(X;Y)\leq r$ represents the compression rate constraint.
\begin{remark}
We define the utility–fairness–compression trade-off, subject to bounded mutual information and compression constraints, as follows.
\begin{align}g^{r}_{\epsilon}(P_{S,X,T})
    &\triangleq\sup_{\begin{array}{c} 
	\substack{P_{Y|X}:S-X-Y,\ T-X-Y,\\I(Y;S)\leq \epsilon^2,\\ I(X;Y)\leq r }
	\end{array}}I(Y;T). \label{prob2}\end{align}
    Then, by using \cite[Remark 3]{khodam}, we have $g^{r}_{\epsilon,\chi^2}(P_{S,X,T})\leq g^{r}_{\epsilon}(P_{S,X,T})$. In other words, one benefit of solving \eqref{problem} is to find a lower bound on \eqref{prob2} which is hard to solve. The trade-off in \eqref{prob2} is equivalent to the privacy-utility trade-off with a rate constraint studied in \cite{gun} and it has been shown that for sufficiently large $r$ ($r \geq H(X)$), the problem leads to a linear program.
    Finally, if we remove the Markov chains $S-X-Y$ and $T-X-Y$ in \eqref{problem}, we get the problem studied in \cite{9965856}. 
\end{remark}
\begin{remark}
	In this work, we assume that $P_{S|X}$ is invertible; however, this assumption can be generalized using the techniques in \cite{Khodam22}. Here, we focus on invertible case for design simplicity and leave the generalization to future work.
\end{remark}

 \section{Main Results}\label{sec:resul}
In this section, we first derive an approximation to \eqref{problem}. We then find lower bounds for this approximation and discuss their tightness.
To this end,
using \eqref{local1}, we can rewrite the conditional distribution $P_{S|Y=y}$ as a perturbation of $P_S$. Hence, for any $y\in\mathcal{Y}$, we can write $P_{S|Y=y}=P_S+\epsilon\cdot J_{y}$, where $J_y\in\mathbb{R}^\mathcal{S}$ is a perturbation vector that has the following three properties
		\begin{align}
	&\sum_s J_{y}(s)=0,\ \ \forall y,\label{0}	 \\
	&\sum_y P_{Y}(y)J_{y}(s)=0, \ \forall s,\label{sum}	\\
	&\|[\sqrt{P_S}^{-1}]J_y\|^2=\sum_x \frac{J_{y}(s)^2}{P_S(s)}\leq 1 \label{1}.
	\end{align} 
    The first two properties ensure that $P_{S|Y=y}$ is a valid probability distribution \cite{khodam,Khodam22}, and the third property follows from \eqref{local1}.
Furthermore, \eqref{sum} implies $\sum_u P_Y(y)J_y=\bm{0}\in\mathbb{R}^{|\mathcal{S}|}$. 
Next, we approximate $I(Y;X)$ and $I(Y;T)$ using the Markov chains and the perturbation vector $J_y$.
In the next results, we use the Bachmann-Landau notation where $o(\epsilon)$ describes the asymptotic behavior of a function $f:\mathbb{R}^+\rightarrow\mathbb{R}$ which satisfies $\frac{f(\epsilon)}{\epsilon}\rightarrow 0$ as $\epsilon\rightarrow 0$.
\begin{lemma}\label{lem1}
	For all $\epsilon<\frac{|\sigma_{\text{min}}(P_{S|X})|\min_{x\in\mathcal{X}}P_X(x)}{\sqrt{\max_{s\in{\mathcal{S}}}P_S(s)}}$, we have
	\begin{align}
I(X;Y)&=\frac{1}{2}\epsilon^2\sum_y\! P_y\|[\sqrt{P_X}^{-1}]P_{S|X}^{-1}J_y\|^2+o(\epsilon^2)\nonumber\\&\simeq \frac{1}{2}\epsilon^2\sum_y P_y\|[\sqrt{P_X}^{-1}]P_{S|X}^{-1}J_y\|^2,\label{appxy}
\end{align}
and for all $\epsilon<\frac{\min_{t\in\mathcal{T}}P_T(t)}{|\sigma_{\text{max}}(P_{T|X}P_{S|X}^{-1})|\sqrt{\max_{s\in{\mathcal{X}}}P_S(s)}}$, 
    we have
\begin{align}
I(T;Y)&=\frac{1}{2}\epsilon^2\sum_y P_y\|[\sqrt{P_T}^{-1}]P_{T|X}P_{S|X}^{-1}J_y\|^2+o(\epsilon^2)\nonumber\\&\simeq \frac{1}{2}\epsilon^2\sum_y P_y\|[\sqrt{P_T}^{-1}]P_{T|X}P_{S|X}^{-1}J_y\|^2.\label{appty}
	\end{align}
\end{lemma}
\begin{proof}
To approximate $I(X;Y)$, using the Markov chain $S-X-Y$ and invertible $P_{S|X}$ we can rewrite $P_{X|Y=y}$ as perturbations of $P_X$ as follows
	\begin{align*}
	P_{X|Y=y}&=P_{S|X}^{-1}[P_{S|Y=y}-P_S]+P_X\\&=\epsilon\cdot P_{S|X}^{-1}J_y+P_X.
	\end{align*}
    Then, by using local approximation of the KL-divergence which is based on the second order Taylor expansion of $\log(1+x)$ we get
    \begin{align*}
	I(Y;X)&=\sum_y P_Y(y)D(P_{X|Y=y}||P_X)\\&\!\!\!\!\!\!\!\!\!\!\!\!\!\!=\sum_y P_Y(y)\sum_x\! P_{X|Y=y}(x)\log\left(\!1\!+\!\epsilon\frac{P_{S|X}^{-1}J_y(x)}{P_X(x)}\right)\\&\!\!\!\!\!\!\!\!\!\!\!\!\!\!=\frac{1}{2}\epsilon^2\sum_y P_Y(y)\sum_x
	\frac{(P_{S|X}^{-1}J_y(x))^2}{P_X(x)}+o(\epsilon^2).
	\end{align*}
    For more detail see \cite[Proposition 3]{khodam}.
    Finally, for approximating $I(X;Y)$ we must have $|\epsilon\frac{P_{S|X}^{-1}J_y(x)}{P_x(x)}|<1$ for all $x$ and $y$. A sufficient condition for $\epsilon$ to satisfy this inequality is to have $\epsilon<\frac{|\sigma_{\text{min}}(P_{S|X})|\min_{x\in\mathcal{X}}P_X(x)}{\sqrt{\max_{s\in{\mathcal{S}}}P_S(s)}}$, since in this case we have
    \begin{align*}
&\!\!\!\!\!\!\epsilon^2|P_{S|X}^{-1}J_y(x)|^2\leq\epsilon^2\left\lVert P_{S|X}^{-1}J_y\right\rVert^2\\&\!\!\!\!\!\!\leq\epsilon^2 \sigma_{\max}^2\left(P_{S|X}^{-1}\right)\left\lVert J_y\right\rVert^2\\&\!\!\!\!\!\!\stackrel{(a)}{\leq}\epsilon^2\max_{s\in{\mathcal{S}}}P_S(s)\frac{1}{\sigma^2_{\text{min}}(P_{S|X})}\\&<\min_{x\in\mathcal{X}} P_X^2(x).
\end{align*}
which implies $$|\epsilon\frac{P_{S|X}^{-1}J_y(x)}{P_X(x)}|<1.$$ The step (a) follows from $\sigma_{\max}^2\left(P_{S|X}^{-1}\right)=\frac{1}{\sigma_{\min}^2\left(P_{S|X}\right)}$ and $\|J_y\|^2\leq\max_{s\in{\mathcal{S}}}P_S(s)$. The latter inequality follows from \eqref{1} since we have
\begin{align*}
\frac{\|J_y\|^2}{\max_{s\in{\mathcal{S}}}P_S(s)}&\leq \sum_{s\in\mathcal{S}}\frac{J_y^2(s)}{P_S(s)}\\&\leq 1.
\end{align*}
	To approximate $I(T;Y)$, using the Markov chain $T-X-Y$ we can rewrite $P_{T|Y=y}$ as perturbations of $P_T$ as follows
	\begin{align*}
	P_{T|Y=y}&=P_{T|X}P_{X|y}=\epsilon\cdot P_{T|X}P_{S|X}^{-1}J_y+P_T.
	\end{align*}
	Then, by using the local approximation of the KL-divergence we have
	\begin{align*}
	I(Y;T)&=\sum_y P_Y(y)D(P_{T|Y=y}||P_T)\\&=\sum_y P_Y(y)\sum_t\! P_{T|Y=y}(t)\log\left(1+\epsilon\frac{P_{T|X}P_{S|X}^{-1}J_y(t)}{P_T(t)}\right)\\&=\frac{1}{2}\epsilon^2\sum_y P_Y(y)\sum_t
	\frac{(P_{T|X}P_{S|X}^{-1}J_y)^2}{P_T}+o(\epsilon^2).
	\end{align*}
	 Due to the Taylor expansion, we must have $|\epsilon\frac{P_{T|X}P_{S|X}^{-1}J_y(t)}{P_T(t)}|<1$ for all $t$ and $y$. A sufficient condition for $\epsilon$ to satisfy this inequality is to have $\epsilon<\frac{\min_{t\in\mathcal{T}}P_T(t)}{|\sigma_{\text{max}}(P_{T|X}P_{S|X}^{-1})|\sqrt{\max_{s\in{\mathcal{X}}}P_S(s)}}$, since in this case we have
\begin{align*}
&\epsilon^2|P_{T|X}P_{S|X}^{-1}J_y(t)|^2\leq\epsilon^2\left\lVert P_{T|X}P_{S|X}^{-1}J_y\right\rVert^2\\&\leq\epsilon^2 \sigma_{\max}^2\left(P_{T|X}P_{S|X}^{-1}\right)\left\lVert J_y\right\rVert^2\\&\stackrel{(a)}\leq\epsilon^2\max_{s\in{\mathcal{S}}}P_S(s)\sigma^2_{\text{max}}(P_{T|X}P_{S|X}^{-1})\\&<\min_{t\in\mathcal{T}} P_T^2(t),
\end{align*}
which implies $$|\epsilon\frac{P_{S|T}^{-1}J_y(t)}{P_T(t)}\!|<1.$$ The step (a) follows from $\|J_y\|^2\leq\max_{s\in{\mathcal{S}}}P_S(s)$. \\
\end{proof}
In the next result, we present an approximation on \eqref{problem}. Before stating the next result, let us define \begin{align*}L_y&\triangleq[\sqrt{P_S}^{-1}]J_y\in\mathbb{R}^{|\mathcal{Y}|},\\ W^{T;Y}&\triangleq[\sqrt{P_T}^{-1}]P_{T|X}P_{S|X}^{-1}[\sqrt{P_S}]\in \mathbb{R}^{|\mathcal{T}|\times |\mathcal{Y}|},\end{align*} and \begin{align*}W^{X;Y}\triangleq[\sqrt{P_X}^{-1}]P_{S|X}^{-1}[\sqrt{P_S}]\in \mathbb{R}^{|\mathcal{X}|\times |\mathcal{Y}|}.\end{align*} Furthermore, let $$c_1\triangleq \frac{\min_{t\in\mathcal{T}}P_T(t)}{|\sigma_{\text{max}}(P_{T|X}P_{S|X}^{-1})|\sqrt{\max_{s\in{\mathcal{X}}}P_S(s)}},$$ and $$c_2\triangleq \frac{|\sigma_{\text{min}}(P_{S|X})|\min_{x\in\mathcal{X}}P_X(x)}{\sqrt{\max_{s\in{\mathcal{S}}}P_S(s)}}.$$
\begin{theorem}\label{th1}
    For all $\epsilon<\min\{c_1,c_2\}$, \eqref{problem} can be approximated by the following problem
    \begin{align}\label{approx}
\eqref{problem}\simeq P_2\triangleq\!\!\!\!\!\!\!\!\!\!\!\!\!\!\!\!\max_{\begin{array}{c} 
	\substack{P_{y},L_y:L_y\perp \sqrt{P_S},\\\sum_y P_yL_y=0,\\ \|L_y\|^2\leq 1,\\ \frac{1}{2}\epsilon^2\sum_y P_y\|W^{X;Y}L_y\|^2\leq r}
	\end{array}}\!\!\!\!\!\!\!\frac{1}{2}\epsilon^2\sum_y P_y\|W^{T;Y}L_y\|^2. 
    \end{align}
\end{theorem}
\textbf{Proof:}
	The proof follows by Lemma \ref{lem1}. The constraints $L_y\perp \sqrt{P_S}$, $\sum_y P_yL_y=0$, and $\|L_y\|^2\leq 1$ are followed by \eqref{0}, \eqref{sum}, and \eqref{1}, for more detail see \cite{khodam}.
\begin{remark}
	As $\epsilon$ decreases, the approximation becomes tighter. This follows since using Taylor expansion, as $\epsilon$ decreases the error term becomes smaller.
\end{remark}
\begin{remark}
The compression-rate constraint $$I(X;Y)\le r$$ can be rewritten as
$$
\frac{1}{2}\epsilon^{2}\sum_{y} P_{y}\,\|W^{X;Y} L_{y}\|^{2}\le r - o(\epsilon^{2}).
$$
Since the compression threshold $r$ is typically much larger than the fairness constraint parameter $\epsilon$, we neglect the $o(\epsilon^{2})$ term relative to $r$. We can also derive an upper bound on the error term and use it here; we leave this to future work.
\end{remark}
Next, we present key properties of the matrices $W^{X;Y}$ and $W^{T;Y}$ that we use to derive lower bounds on \eqref{approx}.
\begin{proposition}\label{pos3}
    Each of the matrices $W^{T;Y}$ and $W^{X;Y}$ has singular value $1$, with corresponding singular vector $\sqrt{P_S}$.
\end{proposition}
\begin{proof}
    We have
    \begin{align*}
	&(W^{X;Y})^TW^{X;Y}\sqrt{P_S}\\&=[\sqrt{P_S}](P_{S|X}^T)^{-1}[\sqrt{P_X}^{-1}][\sqrt{P_X}^{-1}]P_{S|X}^{-1}[\sqrt{P_S}]\sqrt{P_S}\\
	&=[\sqrt{P_S}](P_{S|X}^T)^{-1}[\sqrt{P_X}^{-1}][\sqrt{P_X}^{-1}]P_X\\
	&=[\sqrt{P_S}](P_{S|X}^T)^{-1}\bm{1}=[\sqrt{P_S}]\bm{1}\\&=\sqrt{P_S},
	\end{align*}
    where $\bm{1}$ denotes a vector in which all elements are equal to $1$.
    Furthermore, we have
    \begin{align*}
        &(W^{T;Y})^TW^{T;Y}\sqrt{P_S}\\
	&=[\sqrt{P_S}](P_{T|X}P_{S|X}^{-1})^{T}[\sqrt{P_T}^{-1}][\sqrt{P_T}^{-1}]P_T\\
	&=[\sqrt{P_S}](P_{T|X}P_{S|X}^{-1})^{T}\bm{1}=[\sqrt{P_S}]\bm{1}\\&=\sqrt{P_S}.
    \end{align*}
\end{proof}
Next, we find lower bounds on \eqref{approx} and discuss their tightness. To this end, let $\sigma_{\max}^{T;Y}$ and $\sigma_{\max_2}^{T;Y}$ denote the largest and second-largest singular values of the matrix $W^{T;Y}$, with corresponding singular vectors $L_{\sigma}^{W^{T;Y}}$ and $L_{\sigma_2}^{W^{T;Y}}$, respectively, where $\|L_{\sigma}^{W^{T;Y}}\|=1$ and $\|L_{\sigma_2}^{W^{T;Y}}\|=1$.
\begin{theorem}\label{th2}
	If $\sigma_{\max}^{T;Y}>1$, we have
	\begin{align}
	P_2\geq \frac{1}{2}\epsilon^2\left(\frac{\sigma_{\max}^{T;Y}}{K}\right)^2,
	\end{align}
	where $1\leq K$ is the smallest constant that satisfies 
	$$\frac{1}{2}\epsilon^2\|W^{X;Y}L_{\sigma}^{W^{T;Y}}\|^2\leq rK^2.$$ Moreover, when $\frac{1}{2}\epsilon^2\|W^{X;Y}L_{\sigma}^{W^{T;Y}}\|^2\leq r$, then $K= 1$. 
	If $\sigma_{\max}^{T;Y}=1$, then 
	\begin{align}\label{y}
	P_2\geq \frac{1}{2}\epsilon^2\left(\frac{\sigma_{\max_2}^{T;Y}}{K}\right)^2,
	\end{align}
    where $1\leq K$ is the smallest constant that satisfies 
	$$\frac{1}{2}\epsilon^2\|W^{X;Y}L_{\sigma_2}^{W^{T;Y}}\|^2\leq rK^2.$$
	Finally, when $|\mathcal{S}|=2$, if $\sigma_{\max}^{T;Y}>1$, then we have $$P_2= \frac{1}{2}\epsilon^2\left(\frac{\sigma_{\max}^{T;Y}}{K}\right)^2,$$ otherwise, $$P_2= \frac{1}{2}\epsilon^2\left(\frac{\sigma_{\max_2}^{T;Y}}{K}\right)^2.$$
\end{theorem}
\begin{proof}
To derive the first lower bound, let $\sigma_{\max}^{T;Y}>1$. Using Proposition \ref{pos3}, we have $L_{\sigma}^{W^{T;Y}}\perp \sqrt{P_S}$ since $\sqrt{P_S}$ is a singular vector of $W^{T;Y}$. Let $Y$ be a uniform binary RV. Furthermore, let $L_1=\frac{L_{\sigma}^{W^{T;Y}}}{K}$ and $L_2=\frac{-L_{\sigma}^{W^{T;Y}}}{K}$. Clearly, $L_1\perp \sqrt{P_S}$ and $L_2\perp \sqrt{P_S}$ and $P_Y(1)L_1+P_Y(2)L_2=0$. Moreover, we have
\begin{align*}
\frac{1}{2}\epsilon^2\!\!\left(\sum_y \!\!P_Y\|W^{T;Y}L_y\|^2\right)=\frac{1}{2}\epsilon^2\left(\frac{\sigma_{\max}^{T;Y}}{K}\right)^2. 
\end{align*}
The proof of \eqref{y} is similar. To prove the last argument, note that when $|\mathcal{S}|=2$, the matrix $W^{T;Y}$ has two singular vectors, one of which is $\sqrt{P_S}$. Hence, the only feasible direction is the other singular vector which completes the proof.
\end{proof}
\begin{remark}
    In Theorem 2, $K$ is a scaling factor that ensures $L_{\sigma}^{W^{T;Y}}$ is feasible under the compression rate constraint. In other words, if $L_{\sigma}^{W^{T;Y}}$ does not satisfy the rate constraint, we scale it down. We note that if $\tfrac{1}{2}\epsilon^2\|L_{\sigma}^{W^{T;Y}}\|^2 \leq r$, we cannot scale $L_{\sigma}^{W^{T;Y}}$ up, since $\|L_{\sigma}^{W^{T;Y}}\| = 1$. 
\end{remark}
\begin{corollary}
    Let $\epsilon$ be fixed. Then, by using Remark 6, when $r<\tfrac{1}{2}\epsilon^2\|L_{\sigma}^{W^{T;Y}}\|^2 $, we get
$$K=\sqrt{\frac{\frac{1}{2}\epsilon^2\|W^{X;Y}L_{\sigma}^{W^{T;Y}}\|^2}{r}}> 1,$$
    which results
    \begin{align}\label{q}
    P_2\geq r\left(\frac{\sigma_{\max}^{T;Y}}{\|W^{X;Y}L_{\sigma}^{W^{T;Y}}\|^2}\right)^2.
    \end{align}
    When $|\mathcal{S}|=2$, \eqref{q} is tight.
    In other words, when $r$ is smaller than a threshold, the lower bound on $P_2$ increases linearly with respect to $r$.
\end{corollary}
After finding $L_u$ and $P_U$ that attain the lower bounds in Theorem \ref{th2}, we can find the joint distribution $P_{STXY}$ as
\begin{align*}
P_{S|Y=0}&=P_S+\epsilon[\sqrt{P_S}]L_1,\\ 
P_{S|Y=1}&=P_S+\epsilon[\sqrt{P_S}]L_2,\\
P_{T|Y=0}&=P_T+\epsilon P_{T|X}P_{S|X}^{-1}[\sqrt{P_S}]L_1,\\
P_{T|Y=1}&=P_T+\epsilon P_{T|X}P_{S|X}^{-1}[\sqrt{P_S}]L_2,\\
P_{X|Y=0}&=P_X+\epsilon P_{S|X}^{-1}[\sqrt{P_S}]L_1,\\
P_{X|Y=1}&=P_X+\epsilon P_{S|X}^{-1}[\sqrt{P_S}]L_2,
\end{align*}
where $L_1$, $L_2$ and marginal distribution of $Y$ are obtained in Theorem \ref{th2}. Finally, $$P_{STXY}(x,y,,s,t)=P_{XST}(x,s,t)P_{Y|X}(y|x),$$ where $P_{Y|X}$ can be calculated by using $P_{X|Y}$ and $P_Y$ obtained in Theorem \ref{th2}.
\begin{figure}[]
	\centering
	\includegraphics[width = 0.7\textwidth]{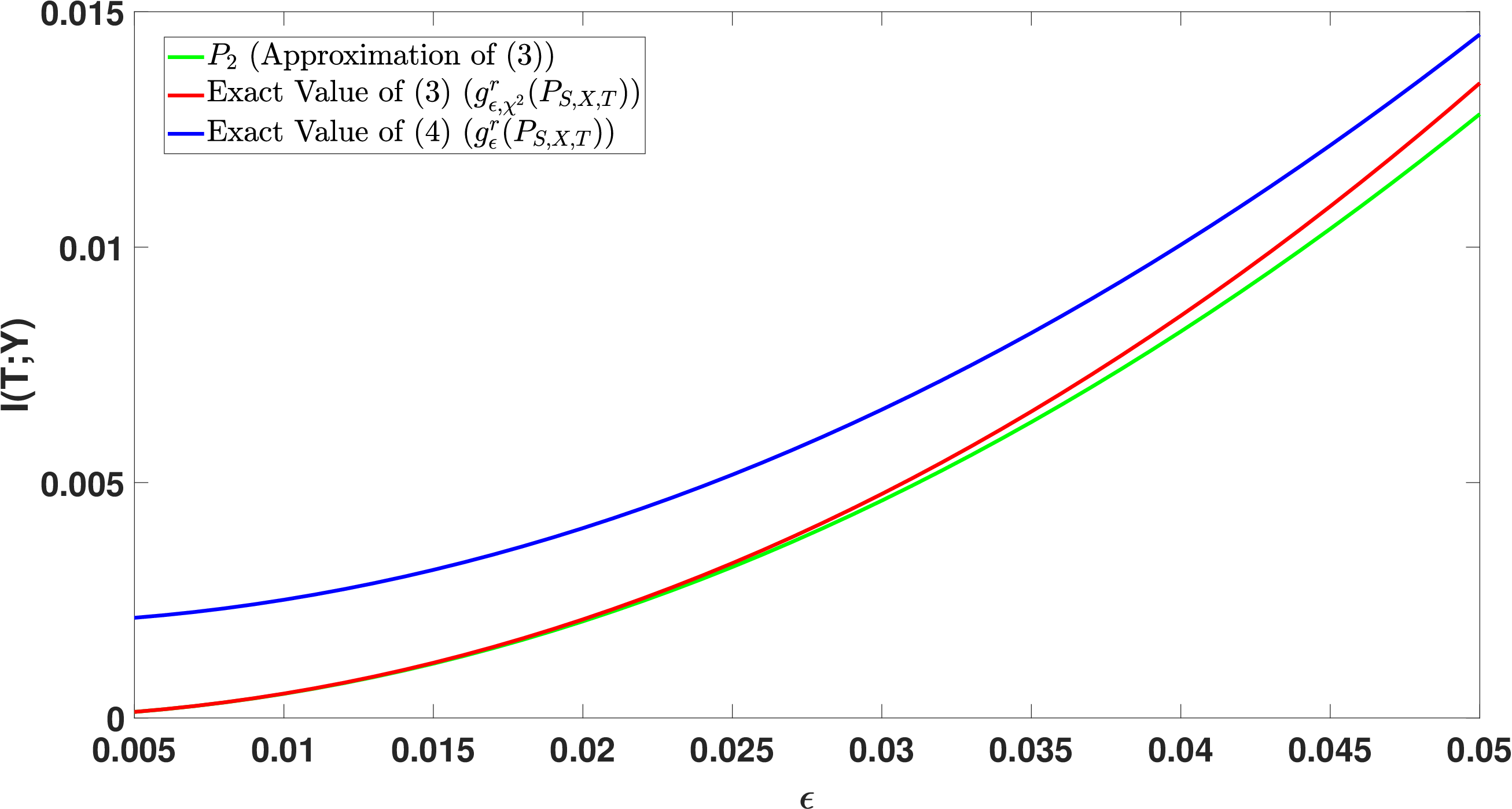}
	\caption{Comparing the proposed approach in this paper with the optimal solutions of \eqref{problem} and \eqref{prob2}. We can see that in the high privacy regimes, $P_2$ is close to the exact solution of \eqref{problem} which is found by exhaustive search and $g^{r}_{\epsilon}(P_{S,X,T})$ dominates $g^{r}_{\epsilon,\chi^2}(P_{S,X,T})$.}
	\label{geo4}
\end{figure}
\section{Numerical Example:}
In this section, we present a numerical example to evaluate the proposed approach and compare the results with the optimal solutions of \eqref{problem} and \eqref{prob2}.\\
\textbf{Example:}
	Let $P_{T|X}=\begin{bmatrix}
	\frac{1}{4}       & \frac{2}{5}  \\
	\frac{3}{4}    & \frac{3}{5}
	\end{bmatrix},$
	and $P_X$ be given as $[\frac{1}{4} , \frac{3}{4}]^T$. Furthermore, let $P_{S|X}=\begin{bmatrix}
	0.275       & 0.32  \\
	0.725    & 0.68
	\end{bmatrix}$. Thus, we find $P_T=P_{T|X}P_X=[0.3625, 0.6375]^T$ and $P_S=P_{S|X}P_X=[0.3088, 0.6913]^T$. Furthermore, we find $$W^{T;Y}=\begin{bmatrix}
	2.4610    & -0.9206  \\
	-1.1599    &  1.7355 
	\end{bmatrix},$$ and $$W^{X;Y}=\begin{bmatrix}
	-16.7931    & 11.8246  \\
	-10.3371    &  -5.8669 
	\end{bmatrix}.$$
	The singular values of $W^{T;Y}$ are $3.2034$ and $1$ with corresponding right singular vectors $$[-0.8314, 0.5557]^T,$$ and $$[0.5557 , 0.8314]^T,$$ respectively.  Moreover, the
    singular values of $W^{X;Y}$ are $23.7087$ and $1$ with corresponding right singular vectors $$[0.8314, -0.5557]^T,$$ and $$[0.5557 , 0.8314]^T,$$ respectively. Here, we let $\epsilon\in [0.005, 0.05]$ and $r=0.75$. We have $$\|W^{X;Y}L_{\sigma}^{W^{T;Y}}\|^2=562.1029,$$ and we can check that for $\epsilon\in [0.005, 0.05]$ the compression rate constraint is satisfied, since for $\epsilon=0.05$, we have $$\frac{1}{2}\epsilon^2\|W^{X;Y}L_{\sigma}^{W^{T;Y}}\|^2= 0.7026.$$
     Here, $L_1$ equals to $[-0.8314, 0.5557]^T$ and the compression rate constraint is satisfied for $L_1$ which yields $K=1$. Using Theorem 2, we have $$P_2=\frac{1}{2}\epsilon^2(3.2034)^2,$$ and the lower bound is tight since $|\mathcal{S}|=2$.
In Fig. \ref{geo4}, we compare $P_2$ and the optimal solutions of \eqref{problem} and \eqref{prob2} using exhaustive search. We recall that $P_2$ is the approximate of \eqref{problem} and here achieves the lower bound in Theorem 2. We can see that $g^{r}_{\epsilon}(P_{S,X,T})$ dominates $g^{r}_{\epsilon,\chi^2}(P_{S,X,T})$, and the gap between exact values of \eqref{problem} and \eqref{prob2} is decreasing as $\epsilon$ increases. Finally, the gap between $P_2$ and $g^{r}_{\epsilon,\chi^2}(P_{S,X,T})$ is small in the high privacy regimes. Intuitively, the blue curve dominates \eqref{problem}, since in \eqref{problem} we impose the point-wise $\chi^2$-criterion instead of bounded mutual information. Furthermore, when $r \geq H(X) = 0.8113$, we can use the linear program from \cite{gun} to obtain the optimal solution to \eqref{prob2}. Finally, let us fix $\epsilon=0.05$. Then, for $r\leq 0.7026$, we can use \eqref{q}.
\section{conclusion}\label{concul}
We have shown that concepts from information geometry can be used to simplify an information-theoretic fair mechanism design problem under the bounded point-wise statistical parity. When a small $\epsilon$ is allowed, simple approximate solutions are derived. Specifically, we look for vectors satisfying the fairness constraints of having the largest Euclidean norm, leading to finding the largest principle singular value and vector of a matrix. The proposed approach establishes a useful and general design framework, which has been used in other privacy design problems such as considering point-wise $\chi^2$ and $\ell_1$-criterion in the literature. 
\clearpage   
\bibliographystyle{IEEEtran}
{\balance \bibliography{IEEEabrv,IZS2026}}

\begin{thebibliography}{10}
\providecommand{\url}[1]{#1}
\csname url@samestyle\endcsname
\providecommand{\newblock}{\relax}
\providecommand{\bibinfo}[2]{#2}
\providecommand{\BIBentrySTDinterwordspacing}{\spaceskip=0pt\relax}
\providecommand{\BIBentryALTinterwordstretchfactor}{4}
\providecommand{\BIBentryALTinterwordspacing}{\spaceskip=\fontdimen2\font plus
\BIBentryALTinterwordstretchfactor\fontdimen3\font minus
  \fontdimen4\font\relax}
\providecommand{\BIBforeignlanguage}[2]{{%
\expandafter\ifx\csname l@#1\endcsname\relax
\typeout{** WARNING: IEEEtran.bst: No hyphenation pattern has been}%
\typeout{** loaded for the language `#1'. Using the pattern for}%
\typeout{** the default language instead.}%
\else
\language=\csname l@#1\endcsname
\fi
#2}}
\providecommand{\BIBdecl}{\relax}
\BIBdecl

\bibitem{AmirITW2024}
A.~Zamani, B.~Rodr{\'\i}guez-G{\'a}lvez, and M.~Skoglund, ``On information
  theoretic fairness: Compressed representations with perfect demographic
  parity,'' in \emph{2024 IEEE Information Theory Workshop (ITW)}, 2024, pp.
  25--30.

\bibitem{vari}
B.~Rodr{\'\i}guez-G{\'a}lvez, R.~Thobaben, and M.~Skoglund, ``A variational
  approach to privacy and fairness,'' in \emph{2021 IEEE Information Theory
  Workshop (ITW)}.\hskip 1em plus 0.5em minus 0.4em\relax IEEE, 2021, pp. 1--6.

\bibitem{gun}
S.~{Sreekumar} and D.~{G\"{u}nd\"{u}z}, ``Optimal privacy-utility trade-off
  under a rate constraint,'' in \emph{2019 IEEE International Symposium on
  Information Theory}, July 2019, pp. 2159--2163.

\bibitem{zhao2022}
H.~Zhao and G.~J. Gordon, ``Inherent tradeoffs in learning fair
  representations,'' \emph{Journal of Machine Learning Research}, vol.~23,
  no.~57, pp. 1--26, 2022.

\bibitem{zhao2019}
H.~Zhao, A.~Coston, T.~Adel, and G.~J. Gordon, ``Conditional learning of fair
  representations,'' in \emph{International Conference on Learning
  Representations}, 2019.

\bibitem{zemel}
R.~Zemel, Y.~Wu, K.~Swersky, T.~Pitassi, and C.~Dwork, ``Learning fair
  representations,'' in \emph{International conference on machine
  learning}.\hskip 1em plus 0.5em minus 0.4em\relax PMLR, 2013, pp. 325--333.

\bibitem{gronowski2023classification}
A.~Gronowski, W.~Paul, F.~Alajaji, B.~Gharesifard, and P.~Burlina,
  ``Classification utility, fairness, and compactness via tunable information
  bottleneck and r{\'e}nyi measures,'' \emph{IEEE Transactions on Information
  Forensics and Security}, 2023.

\bibitem{hardt}
M.~Hardt, E.~Price, and N.~Srebro, ``Equality of opportunity in supervised
  learning,'' \emph{Advances in neural information processing systems},
  vol.~29, 2016.

\bibitem{king3}
A.~Zamani, T.~J. Oechtering, and M.~Skoglund, ``On the privacy-utility
  trade-off with and without direct access to the private data,'' \emph{IEEE
  Transactions on Information Theory}, vol.~70, no.~3, pp. 2177--2200, 2024.

\bibitem{borz}
B.~{Rassouli} and D.~{G\"{u}nd\"{u}z}, ``On perfect privacy,'' \emph{IEEE
  Journal on Selected Areas in Information Theory}, vol.~2, no.~1, pp.
  177--191, 2021.

\bibitem{khodam}
A.~Zamani, T.~J. Oechtering, and M.~Skoglund, ``A design framework for strongly
  $\chi^2$-private data disclosure,'' \emph{IEEE Transactions on Information
  Forensics and Security}, vol.~16, pp. 2312--2325, 2021.

\bibitem{Khodam22}
{A. Zamani, T. J. Oechtering, and M. Skoglund}, ``Data disclosure with non-zero
  leakage and non-invertible leakage matrix,'' \emph{IEEE Transactions on
  Information Forensics and Security}, vol.~17, pp. 165--179, 2022.

\bibitem{Yanina1}
Y.~Y. Shkel, R.~S. Blum, and H.~V. Poor, ``Secrecy by design with applications
  to privacy and compression,'' \emph{IEEE Transactions on Information Theory},
  vol.~67, no.~2, pp. 824--843, 2021.

\bibitem{makhdoumi}
A.~Makhdoumi, S.~Salamatian, N.~Fawaz, and M.~M{\'e}dard, ``From the
  information bottleneck to the privacy funnel,'' in \emph{2014 IEEE
  Information Theory Workshop}, 2014, pp. 501--505.

\bibitem{yamamoto1988rate}
H.~Yamamoto, ``A rate-distortion problem for a communication system with a
  secondary decoder to be hindered,'' \emph{IEEE Transactions on Information
  Theory}, vol.~34, no.~4, pp. 835--842, 1988.

\bibitem{sankar}
L.~Sankar, S.~R. Rajagopalan, and H.~V. Poor, ``Utility-privacy tradeoffs in
  databases: An information-theoretic approach,'' \emph{IEEE Transactions on
  Information Forensics and Security}, vol.~8, no.~6, pp. 838--852, 2013.

\bibitem{Calmon2}
H.~{Wang}, L.~{Vo}, F.~P. {Calmon}, M.~{M\'{e}dard}, K.~R. {Duffy}, and
  M.~{Varia}, ``Privacy with estimation guarantees,'' \emph{IEEE Transactions
  on Information Theory}, vol.~65, no.~12, pp. 8025--8042, Dec 2019.

\bibitem{fairAsoodeh}
H.~Ghoukasian and S.~Asoodeh, ``Differentially private fair binary
  classifications,'' in \emph{2024 IEEE International Symposium on Information
  Theory (ISIT)}, 2024, pp. 611--616.

\bibitem{9518124}
S.~Asoodeh, W.-N. Chen, F.~P. Calmon, and A.~Özgür, ``Differentially private
  federated learning: An information-theoretic perspective,'' in \emph{2021
  IEEE International Symposium on Information Theory (ISIT)}, 2021, pp.
  344--349.

\bibitem{isik2023exact}
B.~Isik, W.-N. Chen, A.~Özgür, T.~Weissman, and A.~No, ``Exact optimality of
  communication-privacy-utility tradeoffs in distributed mean estimation,''
  \emph{Advances in Neural Information Processing Systems}, vol.~36, pp.
  37\,761--37\,785, 2023.

\bibitem{liu2024universal}
Y.~Liu, W.-N. Chen, A.~Özgür, and C.~T. Li, ``Universal exact compression of
  differentially private mechanisms,'' \emph{Advances in Neural Information
  Processing Systems}, vol.~37, pp. 91\,492--91\,531, 2024.

\bibitem{11123370}
A.~Zamani, A.~Changizi, R.~Thobaben, and M.~Skoglund, ``Information-theoretic
  fairness with a bounded statistical parity constraint,'' in \emph{2025 23rd
  International Symposium on Modeling and Optimization in Mobile, Ad Hoc, and
  Wireless Networks (WiOpt)}, 2025, pp. 1--8.

\bibitem{shannon}
C.~E. Shannon, ``Communication theory of secrecy systems,'' \emph{The Bell
  system technical journal}, vol.~28, no.~4, pp. 656--715, 1949.

\bibitem{feldman2015certifying}
M.~Feldman, S.~A. Friedler, J.~Moeller, C.~Scheidegger, and
  S.~Venkatasubramanian, ``Certifying and removing disparate impact,'' in
  \emph{proceedings of the 21th ACM SIGKDD international conference on
  knowledge discovery and data mining}, 2015, pp. 259--268.

\bibitem{5360534}
T.~Calders, F.~Kamiran, and M.~Pechenizkiy, ``Building classifiers with
  independency constraints,'' in \emph{2009 IEEE International Conference on
  Data Mining Workshops}, 2009, pp. 13--18.

\bibitem{Shashi}
S.~Borade and L.~Zheng, ``Euclidean information theory,'' in \emph{2008 IEEE
  International Zurich Seminar on Communications}, 2008, pp. 14--17.

\bibitem{huang}
S.~L. Huang and L.~Zheng, ``Linear information coupling problems,'' in
  \emph{2012 IEEE International Symposium on Information Theory
  Proceedings}.\hskip 1em plus 0.5em minus 0.4em\relax IEEE, 2012, pp.
  1029--1033.

\bibitem{9965856}
A.~Zamani, T.~J. Oechtering, and M.~Skoglund, ``Bounds for privacy-utility
  trade-off with per-letter privacy constraints and non-zero leakage,'' in
  \emph{2022 IEEE Information Theory Workshop (ITW)}, 2022, pp. 13--18.

\bibitem{emma}
E.~M. Athanasakos, N.~Kalouptsidis, and H.~Manjunath, ``Local approximation of
  secrecy capacity,'' in \emph{2024 32nd European Signal Processing Conference
  (EUSIPCO)}, 2024, pp. 745--749.

\bibitem{8125176}
J.~Liao, L.~Sankar, V.~Y.~F. Tan, and F.~du~Pin~Calmon, ``Hypothesis testing
  under mutual information privacy constraints in the high privacy regime,''
  \emph{IEEE Transactions on Information Forensics and Security}, vol.~13,
  no.~4, pp. 1058--1071, 2018.

\bibitem{zamani2025information}
A.~Zamani, P.~Sadeghi, and M.~Skoglund, ``An information geometric approach to
  local information privacy with applications to max-lift and local
  differential privacy,'' \emph{arXiv preprint arXiv:2501.11757}, 2025.

\bibitem{edwards2016censoring}
H.~Edwards and A.~Storkey, ``Censoring representations with an adversary,'' in
  \emph{International Conference on Learning Representations (ICLR)}, 2016, pp.
  1--14.

\bibitem{creager2019flexibly}
E.~Creager, D.~Madras, J.-H. Jacobsen, M.~Weis, K.~Swersky, T.~Pitassi, and
  R.~Zemel, ``Flexibly fair representation learning by disentanglement,'' in
  \emph{International conference on machine learning}.\hskip 1em plus 0.5em
  minus 0.4em\relax PMLR, 2019, pp. 1436--1445.

\bibitem{louizos2015variational}
C.~Louizos, K.~Swersky, Y.~Li, M.~Welling, and R.~Zemel, ``The variational fair
  autoencoder,'' \emph{International Conference on Learning Representations
  (ICLR)}, 2016.

\bibitem{gupta2021controllable}
U.~Gupta, A.~M. Ferber, B.~Dilkina, and G.~Ver~Steeg, ``Controllable guarantees
  for fair outcomes via contrastive information estimation,'' in
  \emph{Proceedings of the AAAI Conference on Artificial Intelligence},
  vol.~35, no.~9, 2021, pp. 7610--7619.

\bibitem{de2022funck}
J.~M. de~Freitas and B.~C. Geiger, ``Funck: Information funnels and bottlenecks
  for invariant representation learning,'' \emph{arXiv preprint
  arXiv:2211.01446}, 2022.

\bibitem{9517766}
S.~Khodadadian, A.~Ghassami, and N.~Kiyavash, ``Impact of data processing on
  fairness in supervised learning,'' in \emph{2021 IEEE International Symposium
  on Information Theory (ISIT)}, 2021, pp. 2643--2648.

\bibitem{zamani2025variable}
A.~Zamani and M.~Skoglund, ``Variable-length coding with zero and non-zero
  privacy leakage,'' \emph{Entropy}, vol.~27, no.~2, p. 124, 2025.

\bibitem{yamamoto}
H.~Yamamoto, ``A source coding problem for sources with additional outputs to
  keep secret from the receiver or wiretappers (corresp.),'' \emph{IEEE
  Transactions on Information Theory}, vol.~29, no.~6, pp. 918--923, 1983.

\end{thebibliography}
\end{document}